\documentclass{article}

\usepackage[square,sort,comma,numbers]{natbib}
\usepackage{relsize}
\usepackage[svgnames,dvipsnames]{xcolor}
\usepackage{amssymb}
\usepackage[T1]{fontenc}
\usepackage{mathtools}
\usepackage{breqn}
\newcommand{\sfrac}[2]{\frac{#1}{#2}}
\usepackage{braket}
\usepackage{physics}
\usepackage{tikz}
\usepackage{xparse}
\usepackage{floatrow}
\usepackage{varwidth}
\usepackage[inline]{enumitem}
\usepackage{bm}
\usepackage[unicode]{hyperref}
\usepackage[margin=1.25in,headheight=13.6pt]{geometry}
\usepackage{tasks}
\usepackage[utf8]{inputenc}

\usepackage{amsthm}

\usepackage[noabbrev,capitalize,nameinlink]{cleveref}
\usepackage{changepage}
\usepackage[normalem]{ulem}
\usepackage{booktabs}
\usepackage{array}
\usepackage{graphicx}
\usepackage{soul}
\usepackage[shortcuts]{extdash}
\usepackage{appendix}
\usepackage{pbox}
\usepackage{sectsty}
\usepackage{microtype}
\usepackage{flafter}

\hypersetup{%
  colorlinks,
  bookmarksnumbered,
  pdfencoding=unicode,
  allcolors=DarkSlateGrey
}

\newcolumntype{L}[1]{>{\raggedright\let\newline\\\arraybackslash\hspace{0pt}}m{#1}}
\newcolumntype{C}[1]{>{\centering\let\newline\\\arraybackslash\hspace{0pt}}m{#1}}
\newcolumntype{M}[1]{>{\centering\let\newline\\\arraybackslash\hspace{0pt}$}m{#1}<{$}}
\newcolumntype{R}[1]{>{\raggedleft\let\newline\\\arraybackslash\hspace{0pt}}m{#1}}

\usepackage{setspace}
\usepackage{eqparbox, etoolbox}
\newlist{rdescription}{description}{1}
\AtBeginEnvironment{rdescription}{%
\setlist[rdescription]{leftmargin =\dimexpr\eqboxwidth{Des}+\labelsep}}%

\SetLabelAlign{parright}{\strut\smash{\parbox[t]\labelwidth{\raggedleft#1}}}

\usetikzlibrary{automata, positioning, decorations.pathmorphing, calc, quotes} 

\tikzset{
    every state/.append style={
        execute at begin node=$,
        execute at end node=$
    },
    initial text = 
}

\usepackage{thmtools}
\usepackage{thm-restate}
\declaretheorem{theorem}
\declaretheorem[sibling=theorem]{lemma,corollary}

\theoremstyle{definition}

\newtheorem{definition}[theorem]{Definition}

\crefname{fact}{fact}{facts}
\Crefname{fact}{Fact}{Facts}

\crefalias{constraint}{equation}
\crefname{constraint}{constraint}{constraints}
\Crefname{constraint}{Constraint}{Constraints}
\creflabelformat{constraint}{#2\textup{(#1)}#3} 
\crefalias{sentence}{equation}
\crefname{sentence}{sentence}{sentences}
\Crefname{sentence}{Sentence}{Sentences}
\creflabelformat{sentence}{#2\textup{(#1)}#3}
\crefalias{expression}{equation}
\crefname{expression}{expression}{expressions}
\Crefname{expression}{Expression}{Expressions}
\creflabelformat{expression}{#2\textup{(#1)}#3}

\NewDocumentEnvironment{delineate}{m}{\textcolor{cyan!70!black!}{> > > > Begin: #1 > > > >}}{\textcolor{red!70!black!}{< < < < End: #1 < < < <}}

\mathchardef\mhyphen="2D

\DeclarePairedDelimiter\paren\lparen\rparen

\newcommand{\oh}[1]{\ensuremath{\mathit{o}\paren*{#1}}}

\makeatletter
\newcommand*{\IfItalicsTF}{%
  \ifx\f@shape\my@test@it
    \expandafter\@firstoftwo
  \else
    \expandafter\@secondoftwo
  \fi
}
\newcommand*{\my@test@it}{it}
\makeatother

\newcommand{\contextsensitivemathrm}[1]{\IfItalicsTF{\mathit{#1}}{\mathrm{#1}}}

\newcommand\machineformat[1]{\ensuremath{\contextsensitivemathrm{#1}}}
\newcommand\langclassformat[1]{\ensuremath{\mathsf{#1}}}

\newcommand{\langformat}[1]{\ensuremath{\mathtt{#1}}}

\NewDocumentCommand{\DTIME}{ o m }{\langclassformat{DTIME\IfValueTF{#1}{\paren[#1]{#2}}{\paren*{#2}}}}

\NewDocumentCommand{\defineautomata}{ m m }{%
    \expandafter\NewDocumentCommand\csname#1fa\endcsname{}{\ensuremath{\machineformat{#1fa}}}%
    \expandafter\NewDocumentCommand\csname rt#1fa\endcsname{}{\ensuremath{\machineformat{rt\mhyphen#1fa}}}%
    \expandafter\NewDocumentCommand\csname o#1fa\endcsname{ o }{\ensuremath{\machineformat{1#1fa\IfValueT{##1}{\paren*{##1}}}}}%
    \expandafter\NewDocumentCommand\csname t#1fa\endcsname{ o }{\ensuremath{\machineformat{2#1fa\IfValueT{##1}{\paren*{##1}}}}}%
    \expandafter\NewDocumentCommand\csname#1tm\endcsname{}{\ensuremath{\machineformat{#1tm}}}%
    \expandafter\NewDocumentCommand\csname #2FA\endcsname{ d<> }{\ensuremath{\langclassformat{#2FA}\IfValueT{##1}{\paren*{##1}}}}%
    \expandafter\NewDocumentCommand\csname RT#2FA\endcsname{ d<> }{\ensuremath{\langclassformat{RT\mhyphen#2FA}\IfValueT{##1}{\paren*{##1}}}}%
    \expandafter\NewDocumentCommand\csname O#2FA\endcsname{ o d<> }{\ensuremath{\langclassformat{1#2FA}\IfValueTF{##1}{\paren*{##1\IfValueT{##2}{, ##2}}}{\IfValueT{##2}{\paren*{##2}}}}}%
    \expandafter\NewDocumentCommand\csname T#2FA\endcsname{ o d<> }{\ensuremath{\langclassformat{2#2FA}\IfValueTF{##1}{\paren*{##1\IfValueT{##2}{, ##2}}}{\IfValueT{##2}{\paren*{##2}}}}}%
}

\defineautomata{d}{D}
\defineautomata{n}{N}
\defineautomata{a}{A}
\defineautomata{p}{P}
\defineautomata{q}{Q}
\defineautomata{qc}{QC}

\newcommand\setneg[1]{\ensuremath{\overline{#1}}}

\newcommand{\PL}{\langformat{PAD}}

\newcommand{\orderedeq}{\ensuremath{EQ}}

\newcommand{\acc}{\ensuremath{\mathrm{acc}}}
\newcommand{\rej}{\ensuremath{\mathrm{rej}}}

\newcommand{\qacc}{\ensuremath{q_{\acc}}}
\newcommand{\qrej}{\ensuremath{q_{\rej}}}
\newcommand{\sacc}{\ensuremath{s_{\acc}}}
\newcommand{\srej}{\ensuremath{s_{\rej}}}

\newcommand{\lend}{\ensuremath{\rhd}}
\newcommand{\rend}{\ensuremath{\lhd}}

\newcommand{\reject}{\emph{reject}}

\setlist{itemsep=0pt}

\newlist{observations}{enumerate}{1}
\setlist[observations]{
    label=\arabic{*}.,
    ref=\arabic{*}
}
\crefname{observationsi}{observation}{observations}
\Crefname{observationsi}{Observation}{Observations}

\newlist{differences}{enumerate}{1}
\setlist[differences]{
    label=\arabic{*}.,
    ref=\arabic{*}
}
\crefname{differencesi}{difference}{differences}
\Crefname{differencesi}{Difference}{Differences}

\newlist{strategies}{enumerate}{1}
\setlist[strategies]{
    label=Strategy \arabic{*}:,
    ref=\arabic{*},
    labelwidth=\widthof{Strategy 1:},
    leftmargin=\parindent+\labelwidth+\labelsep
}
\crefname{strategiesi}{strategy}{strategies}
\Crefname{strategiesi}{Strategy}{Strategies}

\newlist{turingenum}{enumerate}{1}
\setlist[turingenum]{
    noitemsep,
    labelsep=.5em,
    leftmargin=1em+\parindent,
    labelwidth=1em,
    label=(\Roman{*}),
    ref=\Roman{*}
}

\crefname{turingenumi}{Stg.}{Stgs.}
\Crefname{turingenumi}{Stage}{Stages}

\SetLabelAlign{Center}{\hss#1\hss}

\newcommand{\narrowfont}[3]{\scalebox{#1}[1.0]{#3}}
\newcommand{\turinglabelformat}[1]{\narrowfont{0.85}{-20}{\scriptsize#1}}
\newlength{\turinglabelgap}

\ExplSyntaxOn
\DeclareExpandableDocumentCommand{\IfNoValueOrEmptyTF}{mmm}
 {
  \IfNoValueTF{#1}{#2}
   {
    \tl_if_empty:nTF {#1} {#2} {#3}
   }
 }
\ExplSyntaxOff

\NewDocumentEnvironment{turing}{ O{} m m }
 {\IfNoValueOrEmptyTF{#1}{\setlength{\turinglabelgap}{0em}}{\setlength{\turinglabelgap}{0.5em}}\small\begin{enumerate}[labelsep=0pt,align=left,parsep=0pt,leftmargin=\widthof{\turinglabelformat{#1}}+\turinglabelgap, 
 listparindent=0pt] 
  \item[]\ignorespaces#3\\[0.5em]
  \begin{turingenum}[
    nosep,
    align=Center,
    labelwidth=\widthof{\turinglabelformat{#1}},
    labelsep=\turinglabelgap,
    leftmargin=0em 
  ]}
 {\unskip\end{turingenum}\end{enumerate}}

\usepackage{crossreftools}
\makeatletter
\newcommand{\optionaldesc}[3]{%
  \phantomsection
#1\protected@edef\@currentlabel{#1}\protected@edef\cref@currentlabel{%
    [#3][\arabic{#3}][\cref@result]%
    #1%
  }\label{#2}%
}
\makeatother

\NewDocumentCommand{\defineturingitem}{ m m }{%
    \expandafter\NewDocumentCommand\csname#1item\endcsname{ o o m }{\IfValueTF{##1}{\item[\turinglabelformat{\optionaldesc{##1}{\IfValueTF{##2}{##2}{stg:##1}}{turingenumi}}]\begin{adjustwidth}{#2}{0pt}\ignorespaces##3\end{adjustwidth}}{\item[]\begin{adjustwidth}{#2}{0pt}\ignorespaces##3\end{adjustwidth}}}%
}

\defineturingitem{t}{0em}
\defineturingitem{tt}{1em}
\defineturingitem{ttt}{2em}
\defineturingitem{tttt}{3em}
\defineturingitem{ttttt}{4em}

\makeatletter
\def\squiggly{\bgroup \markoverwith{\lower3.9\p@\hbox{\sixly \scalebox{1.2}[0.65]{\char58}}}\ULon}
\makeatother

\def\mysout{\leavevmode\bgroup\def\ULthickness{1pt}\ULdepth=-.4ex\ULset}

\newcommand{\stkout}[1]{\begingroup\ifmmode\text{\mysout{\ensuremath{#1}}}\else\mysout{#1}\fi\endgroup}

\newcommand{\utkanworry}[1]{\textcolor{red!45!black!90}{\ifmmode\smash[b]{\squiggly{#1}}\else\squiggly{#1}\fi}}

\renewcommand{\textvisiblespace}[1][.7em]{%
  \makebox[#1]{%
    \kern.07em
    \vrule height.5ex
    \hrulefill
    \vrule height.5ex
    \kern.07em
  }
}

\DeclareMathAlphabet{\mathsl}{\encodingdefault}{\familydefault}{m}{sl}
\SetMathAlphabet{\mathsl}{bold}{\encodingdefault}{\familydefault}{bx}{sl}

\usepackage{fancyhdr}
\author{A. C. Cem Say\\\href{mailto:say@bogazici.edu.tr}{\texttt{say@bogazici.edu.tr}}}
\title{Short and useful quantum proofs for sublogarithmic-space verifiers}
\date{\small \itshape
  Department of Computer Engineering,
  Bo\u{g}azi\c{c}i University,
  Bebek 34342,
  \.{I}stanbul,
  T\"{u}rkiye
  }

\newcommand{\ie}{i.e.}

\NewDocumentCommand{\padlang}{ e{_} m }{\ensuremath{\PL_{\IfValueT{#1}{#1}}\paren*{#2}}}

\NewDocumentCommand{\padpair}{ s m }{%
    \IfBooleanTF{#1}%
    {\ensuremath{\paren*{\padlang{#2}, \padlang{\setneg{#2}}}}}%
    {\ensuremath{\paren{\padlang{#2}, \padlang{\setneg{#2}}}}}%
}

\NewDocumentCommand{\branch}{ m o }{\makebox{\textsc{[#1]}\IfValueT{#2}{ Branch #2:}}}
\NewDocumentCommand{\branchpr}{ m o }{\branch{\,Probability: \ensuremath{#1}\,}[#2]}
\NewDocumentCommand{\branchperc}{ m o }{\branch{#1\% prob.}[#2]}

\NewDocumentCommand{\vartextvisiblespace}{ O{.7em} O{.7ex} }{%
  \makebox[#1]{%
    \kern.07em
    \vrule height#2
    \hrulefill
    \vrule height#2
    \kern.07em
  }
}

\newcommand{\estring}{\ensuremath{\lambda}}
\newcommand{\blanksymb}{\ensuremath{\text{\vartextvisiblespace}}}
\NewDocumentCommand{\Qpub}{ e{_} }{\ensuremath{Q_{\textnormal{pub}\IfValueT{#1}{,#1}}}}
\NewDocumentCommand{\Qpri}{ e{_} }{\ensuremath{Q_{\textnormal{pri}\IfValueT{#1}{,#1}}}}
\NewDocumentCommand{\Qcom}{ e{_} }{\ensuremath{Q_{\textnormal{com}\IfValueT{#1}{,#1}}}}
\NewDocumentCommand{\bpub}{ e{_} }{\ensuremath{b_{\textnormal{pub}\IfValueT{#1}{,#1}}}}
\NewDocumentCommand{\bpri}{ e{_} }{\ensuremath{b_{\textnormal{pri}\IfValueT{#1}{,#1}}}}

\makeatletter

\def\@testdef #1#2#3{%
  \def\reserved@a{#3}\expandafter \ifx \csname #1@#2\endcsname
 \reserved@a  \else
\typeout{^^Jlabel #2 changed:^^J%
\meaning\reserved@a^^J%
\expandafter\meaning\csname #1@#2\endcsname^^J}%
\@tempswatrue \fi}

\begin{document}

\maketitle

\begin{abstract}
\noindent 
Quantum Merlin-Arthur proof systems are believed to be stronger than both their classical counterparts and ``stand-alone'' quantum computers when Arthur is assumed to operate in $\Omega(\log n)$ space. No hint of such an advantage over classical computation had emerged from research on smaller space bounds, which had so far concentrated on constant-space verifiers.  We initiate the study of quantum Merlin-Arthur systems with space bounds in $\omega(1) \cap o(\log n)$, and exhibit a  problem family $\mathcal{F}$, whose yes-instances have proofs that are verifiable by polynomial-time quantum Turing machines operating in this regime. We show that no problem in $\mathcal{F}$ has proofs that can be verified classically or is solvable by a stand-alone quantum machine in polynomial time if standard complexity assumptions hold. Unlike previous examples of small-space verifiers, our protocols require only subpolynomial-length quantum proofs.

  \vspace{1em}

\end{abstract}

\section{Introduction}\label{sec:intro}

The study of quantum advantage aims to distinguish setups where a model of quantum computation outperforms its classical counterpart from those where no such superiority is exhibited. Shor's celebrated factorization algorithm \cite{S97} provides one hint in this regard, although an unconditional conclusion of quantum advantage would also require a proof that no polynomial-time classical algorithm exists for that problem. Another conjectured difference between the power of classical and quantum machines, formulated as $\mathsf{MA}\subsetneq\mathsf{QMA}$, pertains to their use as verifiers for purported proofs associated with yes-instances in promise problems.   

Both $\mathsf{MA}$ and $\mathsf{QMA}$ are generalizations of the complexity class $\mathsf{NP}$, which is the set of promise problems that have polynomial-time deterministic Turing machines that take an input string $w$ and an additional \textit{certificate} string $c$ and accept if and only if $c$ constitutes a proof that $w$ is a yes-instance of the problem under consideration. When the verifier in the  $\mathsf{NP}$ scenario is substituted by a probabilistic Turing machine, and is allowed to output the wrong verdict with a small probability, one obtains the class $\mathsf{MA}$, whose name is inspired by the mighty but unreliable Merlin, who prepares the purported proof, and the resource-bounded Arthur, who attempts to verify it. Upgrading the verifier to a quantum machine and allowing the certificate to be a quantum state yields the class $\mathsf{QMA}$. \cite{VW15} As mentioned above, it is believed that $\mathsf{QMA}$ is a proper superset of $\mathsf{MA}$. 

The imposition of stricter space bounds on the polynomial-time verifiers in the setups described above has also been considered. Condon \cite{C93} proved that the version of $\mathsf{MA}$ where Arthur is further restricted to use logarithmic space equals $\mathsf{NP}$. Recently, it was also proven \cite{GR23} that the logspace restriction on the (quantum) Arthur leaves the class $\mathsf{QMA}$ unchanged. Note that, when modeling small-space verifiers, one assumes that the ``proof'', which can be polynomially long,\footnote{In both the classical and quantum \cite{MW05,FKLMN16,FR21} cases, logarithmic-length proofs are useless for logspace verifiers, i.e., such a verifier can be replaced by a stand-alone algorithm that runs within $O(\log n)$ space and solves the same problem.} and therefore cannot be stored in its entirety in the verifier's memory, is streamed bit by bit (or qubit by qubit) into the verifier from its original storage location.

Our focus  is on quantum verifiers with sublogarithmic space bounds. Unlike the case with  $s(n)\in\Omega(\log n)$, one is not able to use $s(n)$-space uniform families of quantum circuits and  $s(n)$-space quantum Turing machines as equivalent models of computation in this even stricter regime, and has to adopt a machine-based model. To date,  studies on quantum Merlin-Arthur proof systems in this regard have only considered   verifiers using  $O(1)$ space. \cite{VY14,Y22}  No hint of any advantage of two-way quantum  finite automata over their classical counterparts has emerged from that work, and all example protocols presented  by Villagra and Yamakami (\cite{VY14}, see also Lemma 5.2 in \cite{NY09}) involve languages which also have Merlin-Arthur systems with classical constant-space verifiers.\footnote{When one considers \textit{interactive} proof systems, where Merlin and Arthur can exchange several messages, constant-space quantum verifiers allowed to run for exponential time outperform their classical counterparts. \cite{DS92,Y13,SY17}} \cite{SY14}

In this paper, we initiate the study of quantum Merlin-Arthur systems with space bounds in $\omega(1) \cap o(\log n)$. In Section \ref{sec:defs}, we provide the relevant background information for our results and define a quantum Turing machine model suitable for representing verifiers under these extreme space bounds. Section \ref{sec:famil} describes our main result, namely,  a  problem family $\mathcal{F}$, whose members have yes-instances with proofs that are verifiable by polynomial-time quantum Turing machines operating in $\omega(\log \log n) \cap o(\log n)$ space. 
We obtain the members of $\mathcal{F}$ by padding $\mathsf{QMA}$-complete problems in a fashion that is decodable in small space.
We show that no problem in $\mathcal{F}$ has proofs that can be verified classically or is solvable by a stand-alone quantum machine in polynomial time if a standard complexity assumption about the difficulty of $\mathsf{QMA}$-complete problems holds. Unlike previous examples \cite{SY14,VY14} of small-space verifiers, our protocols require only subpolynomial-length quantum proofs.  Section \ref{sec:conc} is a conclusion.

\section{Definitions and background}\label{sec:defs}

Early work on quantum Merlin-Arthur systems \cite{KSV02,MW05} did not consider the imposition of subpolynomial space bounds on the polynomial-time verifiers.
The study of space-bounded quantum computation involves additional difficulties which do not manifest themselves when one merely focuses on time complexity. A case in point is the principle of deferred measurement. \cite{NC00} While it has been known for decades that making intermediate measurements during the execution of a quantum algorithm provides no time complexity advantage, this equivalence was only recently proved \cite{FR21} for ($\Omega(\log n)$) space complexity.

One complication that arises when sublogarithmic space bounds are considered is the breakdown of the equivalence between the uniform circuit family (UCF) and Turing machine (TM) models. 
Many works on quantum computation with $\Omega(\log n)$ space are presented in terms of the quantum circuit model, which is favored by quantum complexity theorists and physicists. We shall see shortly that this framework is not suitable for modeling computation with smaller space bounds. Before proceeding with the definition of the TM model to be used in this study, we will give an overview of the UCF-based logarithmic-space verification model of  Gharibian and Rudolph \cite{GR23}. The fact that TM's are able to simulate UCF's with no space overhead   will enable us to use a scaled-down version of a protocol from   \cite{GR23} in our main result.

Gharibian and Rudolph's definition \cite{GR23} of a proof system with a verifier using $s(n)$ space for inputs of length $n$ stipulates the existence of an $s(n)$-space uniform family of quantum circuits realizing the verification. When specialized to the case $s(n)\in\Theta(\log n)$, this requires a deterministic logspace Turing machine that prints circuit descriptions with $n^{O(1)}$  gates (from a fixed set that is universal for quantum computation) for any input of length $n$. These circuits operate on three registers named the \textit{input}, \textit{ancilla}, and \textit{proof} registers, containing $n$, $s(n)$, and $n^{O(1)}$ qubits, respectively. Every such circuit $Q_n$ is structured  so that it does not alter the contents of the input register, and the proof is ``read'' in a streaming fashion in the following manner: For a proof of $m$ qubits, the computation modeled by the circuit is sliced into $m+1$ stages, named $V_0$ through $V_m$. No gate in any $V_i$ subcircuit accesses the proof register. For each $i\in\Set{0,\dots,m-1}$, $Q_n$ contains a gate that swaps the first ancilla qubit with the $(i+1)$st proof qubit between the subcircuits $V_i$ and $V_{i+1}$.
Since the input register is never changed and each proof qubit is accessed only once,  the only ``real'' memory used by the verifier is the collection of logarithmically many qubits (all initialized to $\ket{0}$) in the ancilla register. It is within this framework that 
Gharibian and Rudolph have 
proven  that all problems in $\mathsf{QMA}$ have logspace quantum verifiers inspecting streamed quantum proofs in polynomial time.

It is easy to see that the above-mentioned definition of an $s(n)$-space verifier is unsuitable when $s(n)$ is a sublogarithmic function of $n$.
Although a deterministic Turing machine using memory in $\Omega(\log n)$  can output circuit descriptions with polynomially many gates, where all the polynomially many qubits in the circuit are named as needed, a machine with $o(\log n)$ space, which cannot fit even the index of an individual qubit in its worktape, seems too weak for this purpose. For this reason, we formulate a quantum Turing machine model (based on \cite{W03}) that can be used to study any space bound in order to represent Arthur directly in our protocols. 

In addition to a constant-sized central memory represented by a finite set of classical states, every  verifier $V$ to be modeled as a quantum Turing machine will have a finite \textit{quantum register} consisting of a fixed number of qubits, a classical \textit{measurement register} of fixed size, and four one-way infinite tapes:
\begin{itemize}
    \item The \textit{input tape} contains the input string $w$ sandwiched between the special endmarker symbols $\lend$ and $\rend$,  The read-only, two-way input head is located on the left endmarker $\lend$ at the beginning of the computation.
    \item The \textit{classical work tape} is accessed by a dedicated two-way read-write head which is located on its leftmost cell at the beginning of the computation. All cells contain the blank symbol initially.
    \item The \textit{quantum work tape} is accessed by a dedicated two-way  head which is located on its leftmost cell at the beginning of the computation. Each cell of this tape holds a qubit, and all these qubits are in their zero states initially.
    \item The \textit{proof tape} is accessed by a dedicated one-way  head (which cannot move to the left) that is located on its leftmost cell at the beginning of the computation. We assume that the $m$ qubits whose state is a purported proof that the input string $w$ is a yes-instance of the problem associated with $V$ are situated in the first $m$ cells of this tape at the start, and the head cannot move beyond the area occupied by the proof qubits.
\end{itemize}

\begin{definition}\label{def:qtm}
A \emph{quantum verifier}  is an 11-tuple  $\paren*{S,\Sigma,K,T,k,\delta_1,\delta_2,s_0,\tau_0,s_{acc},s_{rej}}$, where
\begin{itemize}
    \item $S$ is the finite set of classical states,
    \item $\Sigma$ is the finite input alphabet, not containing  $\lend$ and $\rend$, 
    \item $K$ is the finite classical work alphabet, including the blank symbol \blanksymb,
    \item $T$ is the finite set of possible values of the measurement register,
    \item $k$, a positive integer, is the number of qubits in the quantum register,
    \item $\delta_{1}: \paren*{S \setminus \Set{s_{acc},s_{rej}}} \times \paren*{\Sigma \cup \Set{\lend, \rend}} \times K \to \Delta$, where $\Delta$ is the set of all selective quantum operations \cite{W03} 
    that have operator-sum representations consisting of matrices of algebraic numbers,\footnote{It is well known \cite{NC00} that using only the members of a fixed, small set of algebraic numbers as the transition amplitudes (i.e., the entries of these matrices)  is sufficient for universal quantum computation.} is the first-stage transition function, 
    \item $\delta_{2}: \paren*{S \setminus \Set{s_{acc},s_{rej}}} \times \paren*{\Sigma \cup \Set{\lend, \rend}} \times K \times T  \to S \times K \times \Set{-1,0,1}^3\times \Set{0,1}$ is the second-stage transition function, which determines the updates to be performed on the classical components of the machine, including the positions of all  heads, and whether a new proof qubit will be swapped in or not, as described below,
    \item $s_0 \in S$ is the initial   state,
    \item $\tau_0 \in T$ is the initial value of the measurement register, and
    \item $s_{acc}, s_{rej}\in S$, such that $s_{acc} \neq s_{rej}$, are the accept and reject states, respectively.
\end{itemize}
\end{definition}

At the start of  computation,  a quantum verifier $V$ is in state $s_0$, the measurement register contains $\tau_0$, all qubits of the finite quantum register  are in their zero states, and the heads are initialized as described above.
Every computational step of $V$ consists of a \textit{first stage}, which operates on the quantum register and the quantum work tape and concludes by updating the measurement register, and  a \textit{second stage}, which can change the classical state, classical work tape content and the  head positions, and read a new proof qubit, based on the presently scanned input and classical work tape symbols and the value in the measurement register: If $V$ is in some non-halting state $s$, with the input head scanning the symbol $\sigma$ and the classical work tape head scanning the symbol $\kappa$, the selective quantum operation 
$\delta_1(s,\sigma,\kappa)$
acts on the quantum register and the qubit presently scanned by the quantum work tape head. 
This action may be a unitary transformation or a measurement, 
resulting in those $k+1$ qubits evolving, and some output $\tau$ being stored in the measurement register with an associated probability. The ensuing second-stage transition  $\delta_2(s,\sigma,\kappa, \tau)=(s',\kappa',d_I,d_C,d_Q,d_P)$ updates the classical state to $s'$, changes the presently scanned classical work tape symbol to $\kappa'$, and moves the input, classical work, quantum work and proof heads in the directions indicated by the increments $d_I$, $d_C$,  $d_Q$, and $d_P$, respectively.
If $d_P=1$, the second-stage transition concludes by swapping the qubit at the newly arrived position of the proof tape head by the first qubit of the finite quantum register.
We assume that $\delta_2$ is restricted so that it  never attempts to move the input head  off the tape area containing $\lend w\rend$, or to move any  work head to the left of the leftmost cell of the corresponding tape.
Computation halts with acceptance (rejection) when $V$ enters the  state $s_{acc}$ ($s_{rej}$).  

A quantum verifier is said to \textit{run within space} $s(n)$ if, on any input of length $n$, at most $s(n)$ cells are scanned on both the classical and quantum work tapes during the computation. 

\begin{definition}\label{def:QMA}
    A promise problem $A=(A_{yes},A_{no})$ is in $\mathsf{QMASPACE}(s(n))$ if there exists a polynomial-time quantum verifier $V$ which runs within space $O(s(n))$, such that, for every input string $w\in \Sigma^*$,
    \begin{itemize}
    \item if $w \in A_{yes}$, there exists an $m$-qubit quantum proof $\ket{\psi}$  such that $V$ accepts with probability at least $\frac{2}{3}$ when started with $\lend w\rend$ on the input tape and
 $m$ qubits whose state is $\ket{\psi}$ on the proof tape, and
 \item if $w \in A_{no}$, $V$ rejects with probability at least $\frac{2}{3}$ when started with $\lend w\rend$ on the input tape, regardless of the content of the proof tape.
\end{itemize}
\end{definition}

Note that Definition \ref{def:QMA} requires the verifier to respect the time and space bounds even when the input string is in $\Sigma^* \setminus (A_{yes}\cup A_{no})$.

For any $s(n)\in \Omega(\log n)$, machines conforming to Definition \ref{def:qtm} can simulate $s(n)$-space uniform circuit families within the same bound: The deterministic algorithm preparing the circuit for the given input size uses only $s(n)$ cells on the classical work tape, and the quantum register and the quantum work tape accommodate the $O(s(n))$ qubits of the simulated ancilla register. Although the overall computation of the quantum verifier as described is not necessarily unitary (because the classical part can evolve irreversibly), we note that the quantum part evolves unitarily if the simulated circuit contains no intermediate measurements. Recalling Fefferman and Remscrim's demonstration \cite{FR21} that intermediate measurements can be eliminated without increasing space or time complexity, we can  restate the relationship of Gharibian and Rudolph's result to the original definition of $\mathsf{QMA}$ in the terminology of Definition \ref{def:QMA}:

\begin{theorem}\label{thm:sqma}
    \cite{KSV02,GR23}
    $\mathsf{QMA}=\mathsf{QMASPACE}(n^{O(1)})=\mathsf{QMASPACE}(O(\log n))$.
\end{theorem}

Analogously to the well-known definition of $\mathsf{NP}$-completeness, we say that a promise problem $A$ is \textit{$\mathsf{QMA}$-complete} if $A$ is both in $\mathsf{QMA}$ and $\mathsf{QMA}$-hard under polynomial-time classical (Karp) reductions. Several interesting problems have been shown to be   $\mathsf{QMA}$-complete. \cite{VW15}

The known relationships between  $\mathsf{QMA}$ and other complexity classes can be summarized as follows. Since quantum machines can simulate their classical counterparts, we have
\begin{equation*}
    \mathsf{MA}\subseteq \mathsf{QMA}.
\end{equation*}
Since Arthur can ignore Merlin's message while solving the problem, 
\begin{equation*}
    \mathsf{BQP}\subseteq \mathsf{QMA}.
\end{equation*}
We also know \cite{VW15} the following inclusions:
\begin{equation}\label{eq:qma}
    \mathsf{QMA}\subseteq \mathsf{PP}\subseteq \mathsf{PSPACE}\subseteq \mathsf{EXPTIME}
\end{equation}

Unfortunately, as is often the case in computational complexity theory, it is not known whether any of the inclusions stated above is strict. We will therefore base our claim about the ``difficulty'' of the problem family to be introduced in Section \ref{sec:famil} on the following  hypothesis, which we believe reflects the orthodox opinion among experts: 

\paragraph{The  exponential time hypothesis for  $\mathsf{QMA}$-complete problems.} Informally, the exponential time hypothesis for  $\mathsf{QMA}$-complete problems (ETH-QMA for short) states that neither a classical Merlin-Arthur proof system nor a stand-alone quantum algorithm can handle any $\mathsf{QMA}$-complete problem in subexponential time. Let $\mathsf{MATIME}(t(n))$ and $\mathsf{BQTIME}(t(n))$ denote respectively the generalizations of the complexity classes $\mathsf{MA}$ and $\mathsf{BQP}$, where the runtime budget of the machines is the (not necessarily polynomial) function $t(n)$. ETH-QMA is then the claim that the set
\[
\mathsf{MATIME}(t(n))\cup\mathsf{BQTIME}(t(n))
\]
contains no $\mathsf{QMA}$-complete problems if $t(n)$ is a subexponential function, i.e. if $t(n)\in 2^{n^{o(1)}}$.\footnote{By Equation \ref{eq:qma}, we know that  $\mathsf{QMA}$ is included in both $\mathsf{MATIME}(t(n))$ and $\mathsf{BQTIME}(t(n))$
for  $t(n)\in 2^{n^{O(1)}}$. Every $\mathsf{QMA}$-complete problem $A$ solvable in time $2^{n^{k}}$ for some constant $k>0$ can be ``padded'' to obtain another  $\mathsf{QMA}$-complete problem $A'$ which can then be solved by a machine whose runtime is proportional to $2^{n^{k/c}}$ for any positive constant $c$. The task of ``unpacking'' the padded strings is easy for machines employing at least logarithmic space.}

\vspace{0.5cm}

We conclude this section by recalling a classical result by Stearns et al. \cite{SHL65} about sublogarithmic-space computation that will be useful in Section \ref{sec:famil}:
For any integer $l\geq 0$, let bin$(l)$ denote the shortest string representing $l$ in the binary notation. The language 
\[
L_S=\{\text{bin}(0)\$\text{bin}(1)\$\text{bin}(2)\$\cdots\$\text{bin}(k) \mid k>0\}
\]
is in the complexity class $\mathsf{SPACE}(\log \log n)$. The proof uses the fact that the rightmost ``segment'' (i.e. the postfix following the rightmost $\$$) of any $n$-symbol member of $L_S$ is of length $m\in\Theta(\log n)$. A polynomial-time Turing machine compares consecutive segments bit by bit, and halts when it detects an inconsistency or when all comparisons are successful. Since each comparison involves  a counter whose value can be at most $m$, the space complexity is $\Theta(\log \log n)$. Yes-instances of the problem family to be defined in the next section will be based on the pattern of $L_S$.

\section{A 
family of problems}\label{sec:famil}

This section presents the description of a  family $\mathcal{F}$ of promise problems and an examination of $\mathcal{F}$'s complexity-theoretic properties. We give a quick overview: Every member of  $\mathcal{F}$ is obtained by padding a different $\mathsf{QMA}$-complete problem. Since the padding is superpolynomially long, yes-instances of problems in $\mathcal{F}$  need correspondingly shorter quantum proofs compared to yes-instances of the unpadded versions. Since the padding is not \textit{too} long, if the ETH-QMA holds, those seemingly simpler problems are still too difficult for stand-alone quantum computers or classical verifiers to handle in polynomial time. All of this would be straightforward  if Arthur did not have to worry about space bounds, but he has to, and  the discussion below concentrates on the construction of verifiers that use only sublogarithmic amounts of space to carry out their protocols.

\subsection{Padded problems}\label{sec:F}

 Consider the language
\begin{equation*}\label{eq:Lu}
L_u=\{\text{bin}(0)\$\text{bin}(1)\$\text{bin}(2)\$\cdots\$\text{bin}(k) \mid k>0, \text{bin}(k)\in 1^+
\}.     
\end{equation*}
 Note that $L_u\subsetneq L_S$ (Section \ref{sec:defs}), such that $L_u$ contains only the strings whose rightmost segments represent some integer of the form $2^m-1$, where $m$ is a positive integer. 
In the following, let $u_m$ denote the member of $L_u$ which ends with the postfix $\$1^m$. 


We are now ready to describe the problem family $\mathcal{F}$. 
$\mathcal{F}$'s alphabet is $\Sigma=(\Set{0,1,\$})\times (\Set{0,1,\#})$. Strings written using $\Sigma$ are to be interpreted as containing two tracks, with the symbol $(\sigma_u,\sigma_l)$ indicating that the upper (resp. lower) track contains the symbol $\sigma_u$ (resp. $\sigma_l$) at that position. 
 
 A string on $\Sigma$ whose upper track equals $u_m$ and whose lower
track is a member of the finite set
\begin{equation}\label{eq:lowtrack}
  \left\{ w \#^{|u_m| - j} \;\middle|\; w = \{0,1\}^j,\ j \leq 2^{\frac{m}{|\text{bin}(m)|
  }
  } \right\}  
\end{equation}
is said to be a \textit{padding of $w$}. For any language $L$, $pad(L)$ is the set of all strings which are paddings of the members of $L$.

\begin{definition}\label{def:F}
$\mathcal{F}$ is the set of all promise problems of the form  
$A' = (A'_{{yes}}, A'_{{no}})$, such that
\[
A'_{{yes}} = {pad}(A_{{yes}}), \quad \text{and}
\]
\[
A'_{{no}} = {pad}(A_{{no}}) \cup \left( \Sigma \setminus {pad}(\{0,1\}^*) \right)
\]
for some $\mathsf{QMA}$-complete promise problem  
$A = (A_{{yes}}, A_{{no}})$ on the alphabet $\{0,1\}$.    
\end{definition}

Essentially, $A'$ is a heavily padded version of $A$, with strings which are not syntactically correct paddings also considered as no-instances.
Although it will be seen that the precise amount of padding can be selected from a wide variety of superpolynomial (but subexponential) functions, we have fixed that parameter to a specific function to simplify the exposition below.

\subsection{The verifier}\label{sec:V}

Consider a member $A'$ of $\mathcal{F}$. Mirroring Definition \ref{def:F}, we describe  a polynomial-time quantum verifier $V$ for $A'$ that  decides whether its input of length $n$ is a well-formed padding of some binary string $w$ and is able to check a quantum proof which claims $w$ is a yes-instance of $A$. 

$V$ starts by running the $O(\log \log n)$-space deterministic algorithm of Stearns et al. \cite{SHL65} that was described in Section \ref{sec:defs}, followed by an easy check to see if the final segment is composed entirely of 1's, to determine whether the upper track of its input is some member $u_m$ of $L_u$. If this check fails, $V$ halts by declaring that the input is a no-instance. Otherwise, $V$ proceeds to check whether the lower track
is of the form in Expression (\ref{eq:lowtrack}). Note that $V$ already has a $O(\log \log n)$-bit representation of the length of the last segment of the upper track in its classical work tape at this point. The  computation and storage of the integer $i=\lceil\frac{m}{\lfloor\log m\rfloor+1}\rceil$ 
is also  performed in $O(\log \log n)$ space.\footnote{Note that $|\text{bin}(m)|=\lfloor\log m\rfloor+1$. (Logarithms are to the base 2.)}
$V$ then writes the number $2^i$ in binary notation in its classical work tape, using $\lceil\frac{m}{\lfloor\log m\rfloor+1}\rceil+1$, 
i.e. $O(\frac{\log n}{\log \log n})$ bits. It then uses this variable as a counter to check whether the lower track is of the form $\Set{0,1}^j \#^*$, where $j\leq 2^i$.

If the input is verified to be a legitimate padding of some binary string $w$, $V$ runs the streaming quantum Merlin-Arthur protocol (which must exist by Theorem \ref{thm:sqma}) to verify the quantum proof's claim that $w$ is a yes-instance of $A$. As noted in Section \ref{sec:defs}, this protocol uses only logarithmically many qubits in terms of the length of $w$, which corresponds to just $O(\frac{\log n}{\log \log n})$ space in the quantum work tape of $V$. Since the runtime of the protocol is bounded by a polynomial in terms of the (already very small) length of $w$, $V$'s runtime as a verifier for proofs about $A'$ is also polynomially bounded.

We note that the quantum proofs read by $V$ are significantly shorter than those considered previously in the study of constant-space Merlin-Arthur systems. Those protocols \cite{SY14,VY14} have proofs whose lengths are at least linear in terms of the input. $V$ requires its proofs to be only polynomially long \cite{GR23} in terms of the length of $w$. Since 
\begin{equation}\label{eq:obs}
    2^{\frac{\log n}{\log \log n}}=n^{\frac{1}{\log \log n}},
\end{equation} 
any power of $|w|$ is necessarily in $o(n)$. We argue that $V$ handles significantly more difficult problems that those of \cite{SY14,VY14} in the next subsection.

\subsection{The superiority of quantum verifiers for $\mathcal{F}$}

In this subsection, we will provide evidence that  quantum verifiers using sublogarithmic space can handle problems that are impossible for both stand-alone (polynomial-space) quantum algorithms and classical verifiers under comparable time bounds. 

\begin{theorem}\label{thm:main}
    If the ETH-QMA holds, neither $\mathsf{BQP}$ nor $\mathsf{MA}$ contains any member of $\mathcal{F}$.
\end{theorem}
\begin{proof}
    Assume that the ETH-QMA holds, and that some member $A'$ of $\mathcal{F}$ is solved in polynomial time by a stand-alone quantum algorithm $M'$. (The proof works verbatim for a verifier  in a classical Merlin-Arthur system playing the role of $M'$ as well.) We show how to build a quantum algorithm $M$ that solves the $\mathsf{QMA}$-complete problem $A$, from which $A'$ was derived, in subexponential time, thereby violating the ETH-QMA.

    $M$ starts by running the deterministic algorithm  in Figure \ref{fig:reduction} to convert its input $w$ to a padding of $w$ in the format described in Section \ref{sec:F}. In that algorithm, Stage 1 takes polynomial time. To see that it is easy to find an $m$ satisfying the equation depicted in Stage 2, consider  searching for a number $d$ (the ``length'' of $m$) that satisfies $2^{d-1}\leq i\times d \leq 2^d-1$, that is, $\frac{2^{d-1}}{d}\leq i \leq \frac{2^d-1}{d}$. The smallest integer $d$ which satisfies $i \leq \frac{2^d-1}{d}$ is guaranteed to also satisfy $\frac{2^{d-1}}{d}\leq i$, since otherwise  we would have $i \leq \frac{2^{d-1}-1}{d-1}$, contradicting the minimality of $d$.\footnote{$\frac{2^{d-1}}{d}\leq \frac{2^{d-1}-1}{d-1}$ for all $d>1$. Very short inputs are handled specially in Stage 1.} Since $d$ is much smaller than $i$, which is itself logarithmic in terms of the length of the input, Stage 2 also takes polynomial time.

\begin{figure}[htb!]
    \caption{The reduction of $A$ to $A'$}
    \label{fig:reduction} 
    \begin{turing}[(RW)]{V}{On input $w$:}
        \titem[1. ]{Store the integer $i$ such that $i=\lceil r \rceil$, where  $|w|=2^r$. (If $|w|<3$, let $i=2$.)}
        \titem[2. ]{Find the smallest integer $m$ satisfying the equation $m=i\times |\text{bin}(m)|$.}
        \titem[3. ]{Print the padding of $w$ as follows:}
        \ttitem[ 3.1]{The upper track  will contain $u_m$, i.e. the member of $L_u$ that ends with the segment $1^m$.}
        \ttitem[ 3.2]{The lower track  will contain the member of $w\#^*$ whose length matches the string in the upper track.
        }
            \end{turing}
\end{figure}

    Since the computation required to determine the next symbol to print during the execution of  Stage 3 is trivial, the runtime of this stage is determined by the length of the string it prints. Noting that the algorithm of Figure \ref{fig:reduction} clearly prints a padding of $w$, we examine the relationship between the length $n$ of the final printed string and the length $j$ of $w$ to bound this runtime. 

    Stage 1 ensures that $j>2^{i-1}$ for all sufficiently long $w$. Since $m\in \Theta(\log n)$, we have 
    \begin{equation*}
     j\in \Theta(n^{\frac{1}{\log \log n}})   
    \end{equation*}
 by the observation (\ref{eq:obs}) at the end of Section \ref{sec:V}. We ``invert'' that relationship as follows: Taking logarithms and rearranging, one gets
\begin{equation}\label{eq:ll}
\log n = \log j \times \log \log n.    
\end{equation}
    Taking logarithms again, we have
\begin{equation}\label{eq:lll}
\log \log n = \log \log j + \log \log \log n.    
\end{equation}
Clearly, $\log \log j < \log \log n$. Since $\frac{1}{2}\log \log n > \log \log \log n$ for all sufficiently large $n$, Equation \ref{eq:lll} yields 
\begin{equation}\label{eq:3lu}
    \log \log j < \log \log n < 2\log \log j
\end{equation}
for all such $n$. Using Equations \ref{eq:ll} and \ref{eq:3lu}, we get
\[
\log n = \Theta(\log j \times \log \log j),    
\]
and conclude that
\begin{equation}\label{eq:son}
n = j^{\Theta(\log \log j)}.     
\end{equation}

It is easy to see from Equation \ref{eq:son} that $n$ is a subexponential function of $j$. The algorithm of Figure \ref{fig:reduction} therefore reduces $A$ to $A'$ in subexponential time. Algorithm $M$ then submits the output of this reduction as the input of the hypothetical polynomial-time algorithm $M'$, and reports the verdict of $M'$ as its own decision. Since any polynomial in $n$ would again be subexponential in $j$ (Equation \ref{eq:son}), $M$ handles problem $A$ in subexponential time, contradicting the ETH-QMA.
\end{proof}

When comparing our verifiers  with stand-alone quantum algorithms, Theorem \ref{thm:main}  indicates that the quantum proofs read by our machines are ``useful''. If the comparison is with classical verifiers,   Theorem \ref{thm:main} is a potential quantum advantage result.

\section{Concluding remarks}\label{sec:conc}

For the first time, we have demonstrated the existence of a task at which sublogarithmic-space quantum Merlin-Arthur systems seem to outperform both their classical counterparts and stand-alone quantum computers, even when those competitors are allowed to use polynomial amounts of space in their computations. Treatment of such small space bounds required a verifier definition based directly on the quantum Turing machine model. The protocol presented in Section \ref{sec:V} is also novel in the sense that no sublinear-length proofs were considered for such small-space verifiers until now. 

As mentioned during the presentation of $\mathcal{F}$ in Section \ref{sec:F}, it is possible to obtain similar families of problems by varying the amount of padding in the definition. We note that the padding should not be selected to be too long: If Expression (\ref{eq:lowtrack}) is modified so that $w$ is too short, e.g., if $|w|\leq \log \log n$ where $n$ is the length of the padding, even a stand-alone ``brute-force'' deterministic algorithm can solve the associated problem in time polynomial in $n$. (Equation \ref{eq:qma}) It is therefore still an open question whether there exists a potential quantum advantage as the one we demonstrate here for  verifiers respecting even smaller space bounds like  $O(1)$.

\section*{Acknowledgments}

The author thanks Ryan O'Donnell, Abuzer Yakary{\i}lmaz and Utkan Gezer for all the useful discussions. 

\bibliographystyle{abbrvnat}

\bibliography{references} 
\end{document}